\documentclass{llncs}

\usepackage{caption}
\usepackage{amsfonts}
\usepackage[fleqn]{amsmath}
\usepackage{amsmath}
\usepackage{bm}
\usepackage{graphicx}
\usepackage{amsthm}
\usepackage{float}
\usepackage{array}
\usepackage{thmtools,thm-restate}
\usepackage{hyperref}

\usepackage{multirow}
\usepackage{hhline}		
\usepackage{diagbox}

\usepackage[usenames,table,dvipsnames,x11names]{xcolor}

\title{Algorithms for Automatic Ranking of Participants and Tasks in an Anonymized Contest \thanks{This paper is the full version of a paper with the same title to appear in 11th International Conference and Workshops on Algorithms and Computation~\cite{JRG17}.}}
\author{Yang Jiao \and R. Ravi \and Wolfgang Gatterbauer}
\institute{
Tepper School of Business\\
Carnegie Mellon University\\ 
5000 Forbes Ave., Pittsburgh, PA 15213\\
}

\newcolumntype{L}[1]{>{\raggedright\let\newline\\\arraybackslash\hspace{0pt}}m{#1}}
\newcolumntype{C}[1]{>{\centering\let\newline\\\arraybackslash\hspace{0pt}}m{#1}}
\newcolumntype{R}[1]{>{\raggedleft\let\newline\\\arraybackslash\hspace{0pt}}m{#1}}

\begin{document}
\maketitle

\thispagestyle{plain}\pagestyle{plain}

\begin{abstract}
We introduce a new set of problems based on the \emph{Chain Editing problem}.
In our version of Chain Editing, we are given a set of anonymous participants and a set of undisclosed tasks that every participant attempts. For each participant-task pair, we know whether the participant has succeeded at the task or not. We assume that participants vary in their ability to solve tasks, and that tasks vary in their difficulty to be solved. In an ideal world, stronger participants should succeed at a superset of tasks that weaker participants succeed at. Similarly, easier tasks should be completed successfully by a superset of participants who succeed at harder tasks. In reality, it can happen that a stronger participant fails at a task that a weaker participants succeeds at. Our goal is to find a \emph{perfect nesting of the participant-task relations} by flipping a minimum number of participant-task relations, implying such a ``nearest perfect ordering'' to be the one that is closest to the truth of participant strengths and task difficulties. Many variants of the problem are known to be NP-hard.

We propose six natural $k$-near versions of the Chain Editing problem and classify their complexity. The input to a $k$-near Chain Editing problem includes an initial ordering of the participants (or tasks) that we are required to respect by moving each participant (or task) at most $k$ positions from the initial ordering. We obtain surprising results on the complexity of the six $k$-near problems: Five of the problems are polynomial-time solvable using dynamic programming, but one of them is NP-hard.
\end{abstract}

\section{Introduction}
\subsection{Motivation}
Consider a contest with a set $S$ of participants who are required to complete a set $Q$ of tasks. Every participant either succeeds or fails at completing each task. The identities of the participants and the tasks are anonymous. We aim to obtain rankings of the participants' strengths and the tasks' difficulties.  
This situation can be modeled by an unlabeled bipartite graph with participants on one side, tasks on the other side, and edges defined by whether the participant succeeded at the task. From the edges of the bipartite graph, we can infer that a participant $a_2$ is stronger than $a_1$ if the neighborhood of $a_1$ is contained in (or is ``nested in'') that of $a_2$. Similarly, we can infer that a task is easier than another if its neighborhood contains that of the other. See Figure~\ref{fig:allContexts} for a visualization of strengths of participants and difficulties of tasks. If all neighborhoods are nested, then this nesting immediately implies a ranking of the participants and tasks.
However, participants and tasks are not perfect in reality, which may result in a bipartite graph with ``non-nested'' neighborhoods. In more realistic scenarios, we wish to determine a ranking of the participants and the tasks when the starting graph is not ideal, which we define formally in Section~\ref{sect:formulation}.

\begin{figure}
\centering
  \includegraphics[scale=.4]{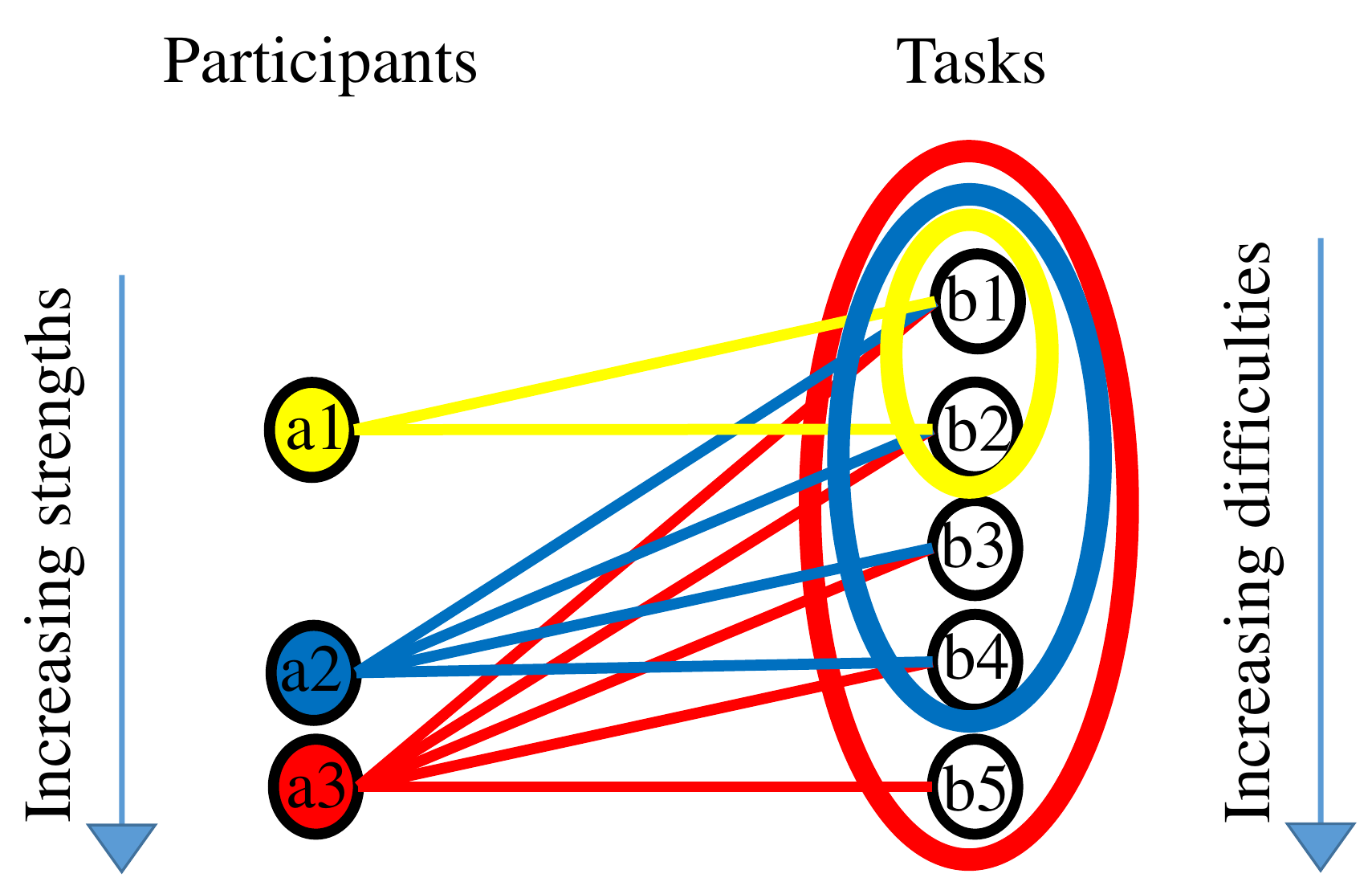}
  \caption{An ideal graph is shown. Participants and tasks may be interpreted as students and questions, or actors and claims. Participant $a_1$ succeeds at $b_1$ to $b_2$; $a_2$ succeeds at $b_1$ to $b_4$; $a_3$ succeeds at $b_1$ to $b_5$. The nesting of neighborhoods here indicate that participant $a_1$ is weaker than $a_2$, who is weaker than $a_3$, and task $b_1$ and $b_2$ are easier than $b_3$ and $b_4$, which in turn are easier than $b_5$.}
  \label{fig:allContexts}
\end{figure}

\subsubsection{Relation to Truth Discovery.}

A popular application of unbiased rankings is computational ``truth discovery.''
\emph{Truth discovery} is the determination of trustworthiness of conflicting pieces of information that are observed often from a variety of sources~\cite{Roth2013:AAAI:tutorial}
and is motivated by the problem of extracting information from networks where the trustworthiness of the actors are uncertain 
\cite{gupta2011heterogeneous}.
The most basic model of the problem is to consider a bipartite graph where one side is made up of actors, the other side is made up of their claims, and edges denote associations between actors and claims.
Furthermore, claims and actors are assumed to have ``trustworthiness'' and ``believability'' scores, respectively, with known a priori values. 
According to a number of recent surveys~\cite{gupta2011heterogeneous,Roth2013:AAAI:tutorial,LGMLSZFH15}, common approaches for truth discovery include
iterative procedures, 
optimization methods, and 
probabilistic graphic models. 
Iterative methods~\cite{DBS09,galland2010corroborating,PasternackRo10,yin2008truth} update 
trust scores of actors to believability scores of claims, and vice versa, until convergence. Various variants of these methods (such as Sums, Hubs and Authorities~\cite{kleinberg1999authoritative}, AverageLog, TruthFinder, Investment, and PooledInvestment) have been extensively studied and proven in practice~\cite{andersen2008trust}.
Optimization methods~\cite{AYLLGD14,LLGZFH14} aim to find truths that minimize the total distance between the provided claims and the output truths for some specified continuous distance function; coordinate descent~\cite{B99} is often used to obtain the solution. 
Probabilistic graphical models~\cite{PasternackRo13} of truth discovery are solved by expectation maximization. 
Other methods for truth discovery include those that leverage trust relationships between the sources~\cite{GatterbauerS2010:ConflictResolution}.
Our study is conceptually closest to optimization approaches (we minimize the number of edge additions or edits), however we suggest a \emph{discrete objective} for minimization, for which
we need to develop new algorithms.

\subsubsection{Our Context: Massively Open Online Courses.}
Our interest in the problem arises from trying to model the problem of automatic grading of large number of students in the context of MOOCs (massively open online courses). Our idea is to crowd-source the creation of automatically gradable questions (like multiple choice items) to students, and have all the students take all questions. 
From the performance of the students, we would like to quickly compute a roughly accurate ordering of the difficulty of the crowd-sourced questions. Additionally, we may also want to efficiently rank the strength of the students based on their performance. 
Henceforth, we refer to participants as students and tasks as questions in the rest of the paper.

\subsubsection{Our Model.}
We cast the ranking problem as a discrete optimization problem of minimizing the number of changes to a given record of the students' performance to obtain nested neighborhoods. This is called the Chain Editing problem. It is often possible that some information regarding the best ranking is already known. For instance, if the observed rankings of students on several previous assignments are consistent, then it is likely that the ranking on the next assignment will be similar. We model known information by imposing an additional constraint that the changes made to correct the errors to an ideal ranking must result in a ranking that is near a given base ranking. By near, we mean that the output position of each student should be within at most $k$ positions from the position in the base ranking, where $k$ is a parameter. Given a nearby ranking for students, we consider all possible variants arising from how the question ranking is constrained. The question ranking may be constrained in one of the following three ways: the exact question ranking is specified (which we term the ``constrained'' case), it must be near a given question ranking (the ``both near'' case), or the question ranking is unconstrained (the ``unconstrained'' case). We provide the formal definitions of these problems next.

\subsection{Problem Formulations}
\label{sect:formulation}
Here, we define all variants of the ranking problem. The basic variants of Chain Editing are defined first and the $k$-near variants are defined afterward.

\subsubsection{Basic Variants of Chain Editing}
\label{sect:basicMOs}
First, we introduce the problem of recognizing an ideal input. Assume that we are given a set $S$ of students, and a set $Q$ of questions, and edges between $S$ and $Q$ that indicate which questions the students answered correctly - note that we assume that every student attempts every question. Denote the resulting bipartite graph by $G = (S \cup Q, E)$. For every pair $(s,q) \in S \times Q$, we are given an edge between $s$ and $q$ if and only if student $s$ answered question $q$ correctly. For a graph $(V,E)$, denote the neighborhood of a vertex $x$ by $N(x):=\{y \in V: xy \in E\}$.

\begin{definition}
We say that student $s_1$ is \emph{stronger} than $s_2$ if $N(s_1) \supset N(s_2)$. We say that question $q_1$ is \emph{harder} than $q_2$ if $N(q_1) \subset N(q_2)$. Given an ordering $\alpha$ on the students and $\beta$ on the questions, $\alpha(s_1) \geq \alpha(s_2)$ shall indicate that $s_1$ is stronger than $s_2$, and $\beta(q_1) \geq \beta(q_2)$ shall indicate that $q_1$ is harder than $q_2$. 
\end{definition}

\begin{definition}
An ordering of the questions satisfies the \emph{interval property} if for every $s$, its neighborhood $N(s)$ consists of a block of consecutive questions (starting with the easiest question) with respect to our ordering of the questions. An ordering $\alpha$ of the students is \emph{nested} if $\alpha(s_1) \geq \alpha(s_2)  \Rightarrow N(s_1) \supseteq N(s_2)$.
\end{definition}

\begin{definition}
The objective of the \emph{Ideal Mutual Orderings (IMO)} problem is to order the students and the questions so that they satisfy the nested and interval properties respectively, or output NO if no such orderings exist.
\end{definition}

Observe that IMO can be solved efficiently by comparing containment relation among the neighborhoods of the students and ordering the questions and students according to the containment order. 

\begin{restatable}{proposition}{IMO}
\label{prop:IMO}
There is a polynomial time algorithm to solve IMO.
\end{restatable}

All missing proofs are in the Appendix~\ref{sect:fullProofs}. Next, observe that the nested property on one side is satisfiable if and only if the interval property on the other side is satisfiable. Hence, we will require only the nested property in subsequent variants of the problem.

\begin{restatable}{proposition}{intervalEquivNested}
\label{prop:intervalEquivNested}
A bipartite graph has an ordering of all vertices so that the questions satisfy the interval property if and only if it has an ordering with the students satisfying the nested property.
\end{restatable}

Next, we define several variants of IMO.

\begin{definition}
In the \emph{Chain Editing (CE)} problem, we are given a bipartite graph representing student-question relations and asked to find a minimum set of edge edits that admits an ordering of the students satisfying the nested property.
\end{definition}

A more restrictive problem than Chain Editing is Chain Addition.
Chain Addition is variant of Chain Editing that allows only edge additions and no deletions. Chain Addition models situations where students sometimes accidentally give wrong answers on questions they know how to solve but never answer a hard problem correctly by luck, e.g. in numerical entry questions.

\begin{definition}
In the \emph{Chain Addition (CA)} problem, we are given a bipartite graph representing student-question relations and asked to find a minimum set of edge additions that admits an ordering of the students satisfying the nested property.
\end{definition}

Analogous to needing only to satisfy one of the two properties, it suffices to find an optimal ordering for only one side. Once one side is fixed, it is easy to find an optimal ordering of the other side respecting the fixed ordering.

\begin{restatable}{proposition}{oneSide}
\label{prop:oneSide}
In Chain Editing, if the best ordering (that minimizes the number of edge edits) for either students or questions is known, then the edge edits and ordering of the other side can be found in polynomial time.
\end{restatable}

\subsubsection{$k$-near Variants of Chain Editing}
We introduce and study the nearby versions of Chain Editing or Chain Addition. Our problem formulations are inspired by Balas and Simonetti's~\cite{BS00} work on $k$-near versions of the TSP. 
\begin{definition}
In the \emph{$k$-near} problem, we are given an initial ordering $\alpha: S \rightarrow [|S|]$ and a positive integer $k$. A \emph{feasible} solution exhibits a set of edge edits (additions) attaining the nested property so that the associated ordering $\pi$, induced by the neighborhood nestings, of the students satisfies $\pi(s) \in [\alpha(s)-k, \alpha(s)+k]$.
\end{definition}
Next, we define three types of $k$-near problems. In the subsequent problem formulations, we bring back the interval property to our constraints since we consider problems where the question side is not allowed to be arbitrarily ordered.

\begin{definition}
In \emph{Unconstrained $k$-near} Chain Editing (Addition), the student ordering must be $k$-near but the question side may be ordered any way. The objective is to minimize the number of edge edits (additions) so that there is a $k$-near ordering of the students that satisfies the nested property.
\end{definition}

\begin{definition}
In \emph{Constrained $k$-near} Chain Editing (Addition), the student ordering must be $k$-near while the questions have a fixed initial ordering that must be kept. The objective is to minimize the number of edge edits (additions) so that there is $k$-near ordering of the students that satisfies the nested property and respects the interval property according to the given question ordering.
\end{definition}

\begin{definition}
In \emph{Both $k$-near} Chain Editing (Addition), both sides must be $k$-near with respect to two given initial orderings on their respective sides. The objective is to minimize the number of edge edits (additions) so that there is a $k$-near ordering of the students that satisfies the nested property and a $k$-near ordering of the questions that satisfies the interval property.
\end{definition}

\subsection{Main Results}
In this paper, we introduce $k$-near models to the Chain Editing problem and present surprising complexity results. Our $k$-near model captures realistic scenarios of MOOCs, where information from past tests is usually known and can be used to arrive at a reliable initial nearby ordering.

We find that five of the $k$-near Editing and Addition problems have polynomial time algorithms while the Unconstrained $k$-near Editing problem is NP-hard. Our intuition is that the Constrained $k$-near and Both $k$-near problems are considerably restrictive on the ordering of the questions, which make it easy to derive the best $k$-near student ordering. The Unconstrained $k$-near Addition problem is easier than the corresponding Editing problem because the correct neighborhood of the students can be inferred from the neighborhoods of all weaker students in the Addition problem, but not for the Editing version.

Aside from restricting the students to be $k$-near, we may consider all possible combinations of whether the students and questions are each $k$-near, fixed, or unconstrained. The remaining (non-symmetric) combinations not covered by the above $k$-near problems are both fixed, one side fixed and the other side unconstrained, and both unconstrained. The both fixed problem is easy as both orderings are given in the input and one only needs to check whether the orderings are consistent with the nesting of the neighborhoods. When one side is fixed and the other is unconstrained, we have already shown that the ordering of the unconstrained side is easily derivable from the ordering of the fixed side via Proposition~\ref{prop:oneSide}. If both sides are unconstrained, this is exactly the Chain Editing (or Addition) problem, which are both known to be NP-hard (see below). Table~\ref{fig:nineCombinations} summarizes the complexity of each problem, including our results for the $k$-near variants, which are starred. Note that the role of the students and questions are symmetric up to flipping the orderings.

\begin{figure}
\centering
\begin{tabular}{ | L{1.2cm} | L{1.3cm} | C{2.4cm} | C{2cm} | C{2cm} | C{2.2cm} | r r}
\hline
\multicolumn{2}{|l|}{\multirow{2}{*}{\diagbox{Questions}{Students}}}
	& \multirow{2}{*}{Unconstrained} 
	& \multicolumn{2}{c |}{$k$-near}	
	& \multirow{2}{*}{Constrained}\\
\multicolumn{2}{|l|}{}	&
	& Editing		
	& Addition	& \\
\hhline{*{6}{-}}
\multicolumn{2}{|l|}{Unconstrained}		
	& \cellcolor{red!30}NP-hard~\cite{DDLS15,Y81} 	
	& \cellcolor{red!30}NP-hard	
	& \cellcolor{green!20}$\mathcal{O}(n^3 k^{2k+2})$	
	& \cellcolor{green!20}$\mathcal{O}(n^2)$\\
\hhline{*{6}{-}}
\multicolumn{1}{|l}{\multirow{2}{*}{$k$-near} }
	& Editing	& \cellcolor{red!30}NP-hard 					
	& \cellcolor{green!20}$\mathcal{O}(n^3 k^{4k+4})$	
	&\cellcolor{black!40}
	&\cellcolor{green!20}$\mathcal{O}(n^3 k^{2k+2})$\\	
\hhline{~|*{5}{-}}
\multicolumn{1}{|l}{}
	& Addition	
	& \cellcolor{green!20}$\mathcal{O}(n^3 k^{2k+2})$ 
	&\cellcolor{black!40}
	&\cellcolor{green!20}$\mathcal{O}(n^3 k^{4k+4})$		
	&\cellcolor{green!20}$\mathcal{O}(n^3 k^{2k+2})$		
	\\
\hhline{*{6}{-}}
\multicolumn{2}{|l|}{Constrained}			
	&\cellcolor{green!20} $\mathcal{O}(n^2)$		
	&\cellcolor{green!20} $\mathcal{O}(n^3 k^{2k+2})$ 
	&\cellcolor{green!20} $\mathcal{O}(n^3 k^{2k+2})$		
	&\cellcolor{green!20}$\mathcal{O}(n^2)$ \\
\hhline{*{6}{-}}
\end{tabular}
\caption{All variants of the problems are shown with their respective complexities. 
The complexity of Unconstrained/Unconstrained Editing~\cite{DDLS15} and Addition~\cite{Y81} were derived before. All other results are given in this paper. 
Most of the problems have the same complexity for both Addition and Editing versions. 
The only exception is the Unconstrained $k$-near version where Editing is NP-hard while Addition has a polynomial time algorithm.
}
\label{fig:nineCombinations}
\end{figure}

To avoid any potential confusion, we emphasize that our algorithms are not fixed-parameter tractable algorithms, as our parameter $k$ is not a property of problem instances, but rather is part of the constraints that are specified for the outputs to satisfy.

The remaining sections are organized as follows. 
Section~\ref{sect:relatedWork} discusses existing work on variants of Chain Editing that have been studied before. Section~\ref{sect:exacts} shows the exact algorithms for five of the $k$-near problems and includes the NP-hardness proof for the last $k$-near problem. Section~\ref{sect:conclusion} summarizes our main contributions.

\section{Related Work}
\label{sect:relatedWork}
The earliest known results on hardness and algorithms tackled Chain Addition. Before stating the results, we define a couple of problems closely related to Chain Addition. The \emph{Minimum Linear Arrangement} problem considers as input a graph $G=(V,E)$ and asks for an ordering $\pi: V \rightarrow [|V|]$ minimizing $\sum_{vw \in E} |\pi(v)-\pi(w)|$. The \emph{Chordal Completion} problem, also known as the \emph{Minimum Fill-In} problem, considers as input a graph $G=(V,E)$ and asks for the minimum size set of edges $F$ to add to $G$ so that $(V,E \cup F)$ has no chordless cycles. A \emph{chordless} cycle is a cycle $(v_1,\ldots, v_r, v_1)$ such that for every $i, j$ with $|i - j| > 1$ and $\{i,j\} \neq \{1,r\}$, we have $v_i v_j \notin E$. Yannakakis~\cite{Y81} proved that Chain Addition is NP-hard by a reduction from Linear Arrangement. He also showed that Chain Addition is a special case of Chordal Completion on graphs of the form $(G=U \cup V,E)$ where $U$ and $V$ are cliques. Recently, Chain Editing was shown to be NP-hard by Drange et. al.~\cite{DDLS15}.

Another problem called \emph{Total Chain Addition} is essentially identical to Chain Addition, except that the objective function counts the number of total edges in the output graph rather than the number of edges added. For Total Chain Addition, Feder et. al.~\cite{FMT09} give a $2$-approximation. The total edge addition version of Chordal Completion has an $O(\sqrt{\Delta}\log^4(n))$-approximation algorithm~\cite{AKR93} where $\Delta$ is the maximum degree of the input graph. For Chain Addition, Feder et. al.~\cite{FMT09} claim an $8d+2$-approximation, where $d$ is the smallest number such that every vertex-induced subgraph of the original graph has some vertex of degree at most $d$. 
Natanzon et. al.~\cite{NSS00} give an $8OPT$-approximation for Chain Addition by approximating Chordal Completion.
However, no approximation algorithms are known for Chain Editing.

Modification to chordless graphs and to chain graphs have also been studied from a fixed-parameter point of view. A \emph{fixed-parameter tractable (FPT)} algorithm for a problem of input size $n$ and parameter $p$ bounding the value of the optimal solution, is an algorithm that outputs an optimal solution in time $O(f(p)n^c)$ for some constant $c$ and some function $f$ dependent on $p$. For Chordal Completion, Kaplan et. al.~\cite{KST99} give an FPT in time $O(2^{O(OPT)} + OPT^2 nm)$. Fomin and Villanger~\cite{FV12} show the first subexponential FPT for Chordal Completion, in time $O(2^{O(\sqrt{OPT} \log OPT)} + OPT^2 nm)$. Cao and Marx~\cite{CM16} study a generalization of Chordal Completion, where three operations are allowed: vertex deletion, edge addition, and edge deletion. There, they give an FPT in time $2^{O(OPT \log OPT)} n^{O(1)}$, where $OPT$ is now the minimum total number of the three operations needed to obtain a chordless graph. For the special case of Chain Editing, Drange et. al.~\cite{DDLS15} show an FPT in time $2^{O(\sqrt{OPT} \log OPT)} + \text{poly}(n)$. They also show the same result holds for a related problem called Threshold Editing. 

On the other side, Drange et. al.~\cite{DDLS15} show that Chain Editing and Threshold Editing do not admit $2^{o(\sqrt{OPT})} \text{poly}(n)$ time algorithms assuming the Exponential Time Hypothesis (ETH). For Chain Completion and Chordal Completion, Bliznets et. al.~\cite{BCKMP16} exclude the possibility of $2^{O(\sqrt{n}/\log n)}$ and $2^{O(OPT^{\frac 1 4} / \log^c k)} n^{O(1)}$ time algorithms assuming ETH, where $c$ is a constant. For Chordal Completion, Cao and Sandeep~\cite{CS16} showed that no algorithms in time $2^{O(\sqrt{OPT} - \delta)} n^{O(1)}$ exist for any positive $\delta$, assuming ETH. They also exclude the possibility of a PTAS for Chordal Completion assuming $P \neq NP$. Wu et. al.~\cite{WAPL14} show that no constant approximation is possible for Chordal Completion assuming the Small Set Expansion Conjecture.
Table~\ref{tab:known} summarizes the known results for the aforementioned graph modification problems.

\begin{center}
\captionof{table}{Known Results} \label{tab:known}
  \begin{tabular}{ | C{2.4cm} | C{4.4cm} | C{4.4cm} |}
    \hline
		\qquad & Chordal & Chain\\
		\hline
    Editing & Unknown approximation, FPT~\cite{DBS09} & Unknown approximation, FPT~\cite{DBS09} \\
		\hline
    Addition & $8OPT$-approx ~\cite{NSS00}, FPT~\cite{DBS09} & $8OPT$-approx ~\cite{NSS00}, $8d+2$-approx ~\cite{FMT09}, FPT~\cite{DBS09} \\
		\hline
    Total Addition & $O(\sqrt{\Delta}\log^4(n))$-approx ~\cite{AKR93}, FPT~\cite{DBS09} & 2-approx ~\cite{FMT09}, FPT~\cite{DBS09}\\
    \hline
  \end{tabular}
\end{center}

For the $k$-near problems, we show that the Unconstrained $k$-near Editing problem is NP-hard by adapting the NP-hardness proof for Threshold Editing from Drange et. al.~\cite{DBS09}. The remaining $k$-near problems have not been studied.

\section{Polynomial Time Algorithms for $k$-near Orderings}
\label{sect:exacts}

We present our polynomial time algorithm for the Constrained $k$-near Addition problem and state similar results for the Constrained $k$-near Editing problem, the Both $k$-near Addition and Editing problems, and the Unconstrained $k$-near Addition problem. The algorithms and analyses for the other polynomial time results use similar ideas as the one for Constrained $k$-near Addition. They are provided in detail in the Appendix~\ref{sect:fullProofs}. We also state the NP-hardness of the Unconstrained $k$-near Editing problem and provide the proof in the Appendix~\ref{sect:fullProofs}.

We assume correct orderings label the students from weakest (smallest label) to strongest (largest label) and label the questions from easiest (smallest label) to hardest (largest label). We associate each student with its initial label given by the $k$-near ordering. For ease of reading, we boldface the definitions essential to the analysis of our algorithm.

\begin{theorem}[Constrained $k$-near Editing]
\label{thm:constrainedEditing}
Constrained $k$-near Editing can be solved in time $O(n^3 k^{2k+2})$.
\end{theorem}

\begin{proof}
\sloppy 
Assume that the students are given in $k$-near order $1, \ldots,|S|$ and that the questions are given in exact order $1 \leq \cdots \leq |Q|$.
We construct a dynamic program for Constrained $k$-near Editing. First, we introduce the subproblems that we will consider. Define $\bm{C(i,u_i,U_i,v_{j_i})}$ to be the smallest number of edges incident to the weakest $i$ positions that must be edited such that $u_i$ is in position $i$, $U_i$ is the set of students in the weakest $i-1$ positions, and $v_{j_i}$ is the hardest question correctly answered by the $i$ weakest students. Before deriving the recurrence, we will define several sets that bound our search space within polynomial size of $n=|S|+|Q|$.

\noindent \textbf{Search Space for $U_i$.}
Given position $i$ and student $u_i$, define $\bm{P_{i,u_i}}$ to be the set of permutations on the elements in $\big[\max\{1,i-k\}, \min\{|S|,i+k-1\}\big] \setminus \{u_i\}$. Let $\bm{F_{i,u_i}}:=\Big\{\{\pi^{-1}(1),\ldots, \pi^{-1}(k)\}: \pi \in P_{i,u_i}, \pi(a) \in [a-k,a+k] \forall a \in \big[\max\{1,i-k\},\min\{|S|,i+k-1\}\big]\setminus \{u_i\} \Big\}$. The set $P_{i,u_i}$ includes all possible permutations of the $2k$ students centered at position $i$, and the set $F_{i,u_i}$ enforces that no student moves more than $k$ positions from its label. 
We claim that every element of $F_{i,u_i}$ is a candidate for $U_i \setminus \big[1,\max\{1,i-k-1\}\big]$ given that $u_i$ is assigned to position $i$. To understand the search space for $U_i$ given $i$ and $u_i$, observe that for all $i \geq 2$, $U_i$ already must include all of $\big[1,\max\{1,i-k-1\}\big]$ since any student initially at position $\leq i-k-1$ cannot move beyond position $i-1$ in a feasible solution. If $i=1$, we have $U_1 = \emptyset$. From now on, we assume $i \geq 2$ and treat the base case $i=1$ at the end. So the set $U_i \setminus \big[1,\max\{1,i-k-1\}\big]$ will uniquely determine $U_i$. We know that $U_i$ cannot include any students with initial label $[k+i, |S|]$ since students of labels $\geq k+i$ must be assigned to positions $i$ or later. So the only uncertainty remaining is which elements in $\big[\max\{1,i-k\}, \min\{|S|,i+k-1\}\big] \setminus \{u_i\}$ make up the set $U_i \setminus \big[1,\max\{1,i-k-1\}\big]$. We may determine all possible candidates for $U_i \setminus \big[1,\max\{1,i-k-1\}\big]$ by trying all permutations of $\big[\max\{1,i-k\}, \min\{|S|,i+k-1\}\big] \setminus \{u_i\}$ that move each student no more than $k$ positions from its input label, which is exactly the set $F_{i,u_i}$.

\noindent \textbf{Feasible and Compatible Subproblems.}
Next, we define $\bm{S_i}= \Big\{(u_i,U_i,v_{j_i}) : u_i \in \big[ \max \{1,i-k\}, \min \{|S|,i+k\} \big], U_i \setminus \big[1,\max\{1,i-k-1\}\big] \in F_{i,u_i}, v_{j_i} \in Q \cup \{0\} \Big\}$. The set $S_i$ represents the search space for all possible vectors $(u_i,U_i,v_{j_i})$ given that $u_i$ is assigned to position $i$. Note that $u_i$ is required to be within $k$ positions of $i$ by the $k$-near constraint. So we encoded this constraint into $S_i$. To account for the possibility that the $i$ weakest students answer no questions correctly, we allow $v_{j_i}$ to be in position $0$, which we take to mean that $U_i \cup \{u_i\}$ gave wrong answers to all questions.

\sloppy Now, we define $\bm{R_{i-1,u_i,U_i,v_{j_i}}}:= \{(u_{i-1}, U_{i-1}, v_{j_{i-1}}) \in S_{i-1}: v_{j_{i-1}} \leq v_{j_i}, U_i = \{u_{i-1}\} \cup U_{i-1}\}$. The set $R_{i-1,u_i,U_i,v_{j_i}}$ represents the search space for smaller subproblems that are compatible with the subproblem $(i, u_i, U_i, v_{j_i})$. More precisely, given that $u_i$ is assigned to position $i$, $U_i$ is the set of students assigned to the weakest $i-1$ positions, and $v_{j_i}$ is the hardest question correctly answered by $U_i \cup u_i$, the set of subproblems of the form $(i-1, u_{i-1}, U_{i-1}, v_{j_{i-1}})$ which do not contradict the aforementioned assumptions encoded by $(i, u_i, U_i, v_{j_i})$ are exactly those whose $(u_{i-1}, U_{i-1}, v_{j_{i-1}})$ belongs to $R_{i-1,u_i,U_i,v_{j_i}}$. We illustrate compatibility in Figure~\ref{fig:constrainedDP}.
\begin{figure}
\centering
  \includegraphics[scale=.35]{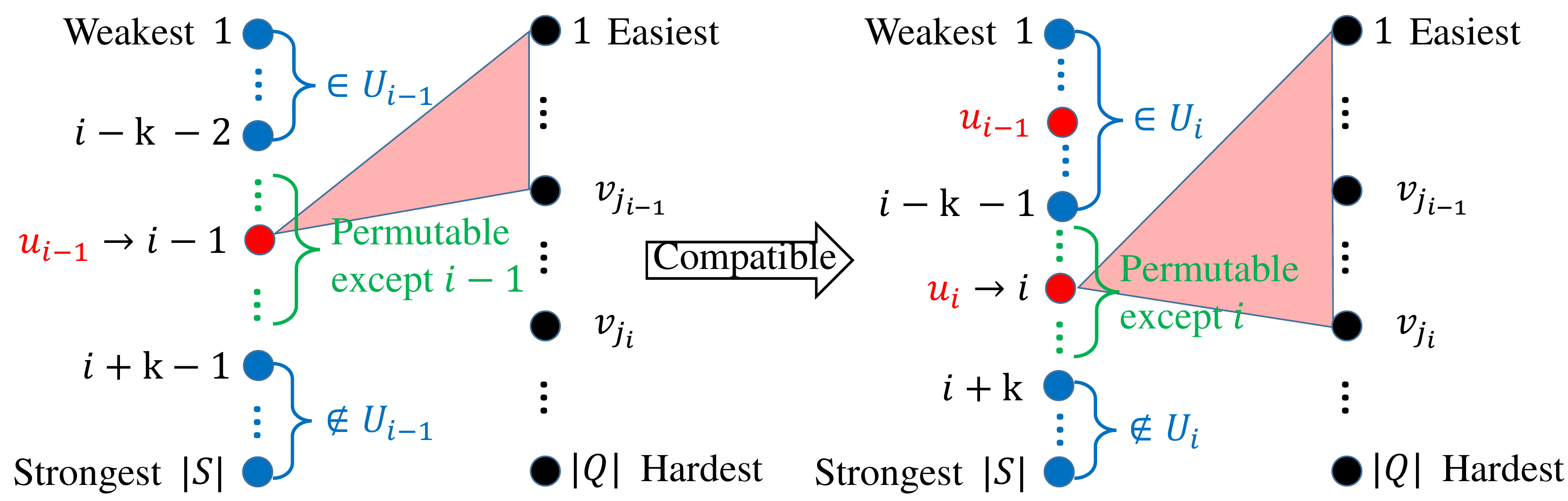}
  \caption{Subproblem $(i-1, u_{i-1}, U_{i-1}, v_{j_{i-1}})$ is compatible with subproblem $(i, u_i, U_i, v_{j_i})$ if and only if $v_{j_{i-1}}$ is no harder than $v_{j_i}$ and $U_i = \{u_{i-1}\} \cup U_{i-1}$. The cost of $(i, u_i, U_i, v_{j_i})$ is the sum of the minimum cost among feasible compatible subproblems of the form $(i-1, u_{i-1}, U_{i-1}, v_{j_{i-1}})$ and the minimum number of edits incident to $u_i$ to make its neighborhood exactly $\{1,\ldots,v_{j_i}\}$.}
  \label{fig:constrainedDP}
\end{figure}

\noindent \textbf{The Dynamic Program.}
Finally, we define $\bm{c_{u_i,v_{j_i}}}$ to be the smallest number of edge edits incident to $u_i$ so that the neighborhood of $u_i$ becomes exactly $\{1, \ldots, v_{j_i}\}$, i.e. $c_{u_i,v_{j_i}} := |N_G(u_i) \triangle \{1, \ldots, v_{j_i}\}|$. We know that $c_{u_i,v_{j_i}}$ is part of the cost within $C(i,u_i,U_i,v_{j_i})$ since $v_{j_i}$ is the hardest question that $U_i \cup \{u_i\}$ is assumed to answer correctly and $u_i$ is a stronger student than those in $U_i$ who are in the positions before $i$. 
We obtain the following recurrence.

$$C(i,u_i,U_i,v_{j_i})=\min_{(u_{i-1},U_{i-1},v_{j_{i-1}}) \in R_{i-1,u_i,U_i,v_{j_i}}} \{C(i-1,u_{i-1},U_{i-1},v_{j_{i-1}})\} + c_{u_i,v_{j_i}}$$

The base cases are $C(1,u_1,U_1,v_{j_1}) = |N_G(u_1) \triangle \{1,\ldots,v_{j_1}\}|$ if $v_{j_1} > 0$, and $C(1,u_1,U_1,v_{j_1}) = |N_G(u_1)|$ if $v_{j_1} = 0$ for all $u_1 \in [1,1+k], v_{j_1} \in Q \cup \{0\}$.

\sloppy By definition of our subproblems, the final solution we seek is $\min_{(u_{|S|},U_{|S|},v_{j_{|S|}}) \in S_{|S|}} C(|S|,u_{|S|},U_{|S|},v_{j_{|S|}})$.

\noindent \textbf{Running Time.}
Now, we bound the run time of the dynamic program. Note that before running the dynamic program, we build the sets $P_{i,u_i}$, $F_{i,u_i}$, $S_i$, $R_{i-1,u_i,U_i,v_{j_i}}$ to ensure that our solution obeys the $k$-near constraint and that the smaller subproblem per recurrence is compatible with the bigger subproblem it came from. Generating the set $P_{i,u_i}$ takes $(2k)!=O(k^k)$ time per $(i,u_i)$. Checking the $k$-near condition to obtain the set $F_{i,u_i}$ while building $P_{i,u_i}$ takes $k^2$ time per $(i,u_i)$. So generating $S_i$ takes $O(k \cdot k^k k^2 \cdot |Q|)$ time per $i$. Knowing $S_{i-1}$, generating $R_{i-1, u_i, U_i, v_{j_i}}$ takes $O(|S|)$ time. Hence, generating all of the sets is dominated by the time to build $\cup_{i \leq |S|} S_i$, which is $O(|S| k^3 k^k |Q|)=O(n^2 k^{k+3})$.

After generating the necessary sets, we solve the dynamic program. Each subproblem $(i,u_i,U_i,v_{j_i})$ takes $O(|R_{i-1,u_i,U_i,v_{j_i}})|$ time. So the total time to solve the dynamic program is $O(\sum_{i \in S, (u_i,U_i,v_{j_i}) \in S_i} |R_{i-1,u_i,U_i,v_{j_i}}|)=O(|S||S_i||S_{i-1}|)=O(n(k \cdot k^k \cdot n)^2)=O(n^3 k^{2k+2})$.
\end{proof}

\begin{restatable}[Constrained $k$-near Addition]{theorem}{constrainedAddition}
\label{thm:constrainedAddition}
Constrained $k$-near Addition can be solved in time $O(n^3 k^{2k+2})$.
\end{restatable}

\vspace{-8pt}

\begin{restatable}[Unconstrained $k$-near Addition]{theorem}{unconstrainedAddition}
\label{thm:unconstrainedAddition}
Unconstrained $k$-near Addition can be solved in time $O(n^3 k^{2k+2})$.
\end{restatable}

\vspace{-8pt}

\begin{restatable}[Unconstrained $k$-near Editing]{theorem}{NPhard}
\label{thm:NPhard}
Unconstrained $k$-near Editing is NP-hard.
\end{restatable}

\vspace{-8pt}

\begin{restatable}[Both $k$-near Editing]{theorem}{bothEditing}
\label{thm:bothEditing}
Both $k$-near Editing can be solved in time $O(n^3 k^{4k+4})$.
\end{restatable}

\vspace{-8pt}

\begin{restatable}[Both $k$-near Addition]{theorem}{bothAddition}
\label{thm:bothAddition}
Both $k$-near Addition can be solved in time $O(n^3 k^{4k+4})$.
\end{restatable}

We present the proofs of the above theorems in the Appendix~\ref{sect:fullProofs}.

\section{Conclusion}
\label{sect:conclusion}
We proposed a new set of problems that arise naturally from ranking participants and tasks in competitive settings and classified the complexity of each problem. First, we introduced six $k$-near variants of the Chain Editing problem, which capture the common scenario of having partial information about the final orderings from past rankings. Second, we provided polynomial time algorithms for five of the problems and showed NP-hardness for the remaining one.\\

\noindent \textbf{Acknowledgments}

\noindent 
This work was supported in part by the US National Science Foundation under award numbers CCF-1527032, CCF-1655442, and IIS-1553547.

\bibliographystyle{splncs03}
\bibliography{database}

\begin{thebibliography}{10}
\providecommand{\url}[1]{\texttt{#1}}
\providecommand{\urlprefix}{URL }

\bibitem{AKR93}
Agrawal, A., Klein, P., Ravi, R.: Cutting down on fill using nested dissection:
  provably good elimination orderings, pp. 31--55. Springer (1993)

\bibitem{andersen2008trust}
Andersen, R., Borgs, C., Chayes, J., Feige, U., Flaxman, A., Kalai, A.,
  Mirrokni, V., Tennenholtz, M.: Trust-based recommendation systems: an
  axiomatic approach. In: WWW. pp. 199--208. ACM (2008)

\bibitem{AYLLGD14}
Aydin, B., Yilmaz, Y., Li, Y., Li, Q., Gao, J., Demirbas, M.: Crowdsourcing for
  multiple-choice question answering. In: IAAI. pp. 2946--2953 (2014)

\bibitem{BS00}
Balas, E., Simonetti, N.: Linear time dynamic-programming algorithms for new
  classes of restricted {TSPs}: A computational study. INFORMS Journal on
  Computing  13(1),  56--75 (2000)

\bibitem{B99}
Bertsekas, D.P.: Non-linear Programming. Athena Scientific (1999)

\bibitem{BCKMP16}
Bliznets, I., Cygan, M., Komosa, P., Mach, L., Pilipczuk, M.: Lower bounds for
  the parameterized complexity of minimum fill-in and other completion
  problems. In: SODA. pp. 1132--1151 (2016)

\bibitem{CM16}
Cao, Y., Marx, D.: Chordal editing is fixed-parameter tractable. Algorithmica
  75(1),  118--137 (2016)

\bibitem{CS16}
Cao, Y., Sandeep, R.B.: Minimum fill-in: Inapproximability and almost tight
  lower bounds. CoRR  abs/1606.08141 (2016),
  \url{http://arxiv.org/abs/1606.08141}

\bibitem{DBS09}
Dong, X.L., Berti-Equille, L., Srivastava, D.: Integrating conflicting data:
  The role of source dependence. PVLDB  2(1),  550--561 (2009)

\bibitem{DDLS15}
Drange, P.G., Dregi, M.S., Lokshtanov, D., Sullivan, B.D.: On the threshold of
  intractability. In: ESA. pp. 411--423 (2015)

\bibitem{FMT09}
Feder, T., Mannila, H., Terzi, E.: Approximating the minimum chain completion
  problem. Inf. Process. Lett.  109(17),  980--985 (2009)

\bibitem{FV12}
Fomin, F.V., Villanger, Y.: Subexponential parameterized algorithm for minimum
  fill-in. In: SODA. pp. 1737--1746 (2012)

\bibitem{galland2010corroborating}
Galland, A., Abiteboul, S., Marian, A., Senellart, P.: Corroborating
  information from disagreeing views. In: WSDM. pp. 131--140. ACM (2010)

\bibitem{GatterbauerS2010:ConflictResolution}
Gatterbauer, W., Suciu, D.: Data conflict resolution using trust mappings. In:
  SIGMOD. pp. 219--230 (2010)

\bibitem{gupta2011heterogeneous}
Gupta, M., Han, J.: Heterogeneous network-based trust analysis: a survey. ACM
  SIGKDD Explorations Newsletter  13(1),  54--71 (2011)

\bibitem{JRG17}
Jiao, Y., Ravi, R., Gatterbauer, W.: Algorithms for automatic ranking of
  participants and tasks in an anonymized contest. In: WALCOM (2017), to appear

\bibitem{KST99}
Kaplan, H., Shamir, R., Tarjan, R.E.: Tractability of parameterized completion
  problems on chordal, strongly chordal, and proper interval graphs. SIAM J.
  Comput.  28(5),  1906--1922 (1999)

\bibitem{kleinberg1999authoritative}
Kleinberg, J.M.: Authoritative sources in a hyperlinked environment. JACM
  46(5),  604--632 (1999)

\bibitem{LLGZFH14}
Li, Q., Li, Y., Gao, J., Zhao, B., Fan, W., Han, J.: Resolving conflicts in
  heterogeneous data by truth discovery and source reliability estimation. In:
  SIGMOD. pp. 1187--1198 (2014)

\bibitem{LGMLSZFH15}
Li, Y., Gao., J., Meng, C., Li, Q., Su, L., Zhao, B., Fan, W., Han, J.: A
  survey on truth discovery. ACM SIGKDD Explorations Newsletter  17(2),  1--16
  (2015)

\bibitem{NSS00}
Natanzon, A., Shamir, R., Sharan, R.: A polynomial approximation algorithm for
  the minimum fill-in problem. SIAM J. Comput.  30(4),  1067--1079 (2000)

\bibitem{PasternackRo10}
Pasternack, J., Roth, D.: Knowing what to believe (when you already know
  something). In: COLING. pp. 877--885 (2010)

\bibitem{PasternackRo13}
Pasternack, J., Roth, D.: Latent credibility analysis. In: WWW. pp. 1009--1021
  (2013)

\bibitem{Roth2013:AAAI:tutorial}
Pasternack, J., Roth, D., Vydiswaran, V.V.: Information trustworthiness. AAAI
  Tutorial (2013)

\bibitem{WAPL14}
Wu, Y.L., Austrin, P., Pitassi, T., Liu, D.: Inapproximability of treewidth,
  one-shot pebbling, and related layout problems. J. Artif. Int. Res.  49(1),
  569--600 (2014)

\bibitem{Y81}
Yannakakis, M.: Computing the minimum fill-in is {NP}-complete. SIAM Journal on
  the Algebraic and Discrete Methods  2(1),  77--79 (1981)

\bibitem{yin2008truth}
Yin, X., Han, J., Yu, P.S.: Truth discovery with multiple conflicting
  information providers on the web. TKDE  20(6),  796--808 (2008)

\end{thebibliography}

 \newpage
 \appendix
 \section{Nomenclature}
 \begin{center}
 \captionof{table}{Overview of notation used in this paper}
   \begin{tabular}{ | l | l | }
     \hline
 		$S$ & set of students \\
     $Q$ & set of questions \\
     $G$ & graph $(S \cup Q, E)$ such that $sq \in E$ iff $s$ answers $q$ correctly\\
     $N(v)$ & neighborhood of $v$ in $G$\\
 		$i$ & position $i$ of a student ranking from weakest to strongest\\
 		$u_i$ & student at position $i$\\
 		$U_i$ & unordered set of $i-1$ students weaker than the $i$th student\\
 		$j_i$ & position of the hardest question correctly answered by the $i$th student\\
 		$v_{j_i}$ & question in position $j_i$ in a question ranking from easiest to hardest\\
 		$V_{j_i}$ & unordered set of questions easier than the $j_i$th question\\
 		\hline
   \end{tabular}
 \end{center}
 \section{Proofs}
 \label{sect:fullProofs}
 In this Section, we present all proofs that were omitted in the main body.

 \IMO*
 \begin{proof}
 \label{pf:IMO}
 Compare the neighborhood of every pair of students $\{s_1, s_2\} \subset S$ and check whether $N(s_1) \subset N(s_2)$ or $N(s_1) \supset N(s_2)$. If $N(s_1) \cap N(s_2)$ is a strict subset of $N(s_1)$ and $N(s_2)$, then output NO. Now, assuming that every pair $\{s_1, s_2\} \subset S$ satisfies $N(s_1) \subset N(s_2)$ or $N(s_1) \supset N(s_2)$, we know that there is an ordering $\alpha: S \rightarrow [|S|]$ such that $\alpha(s_1) \leq \alpha(s_2) \Rightarrow N(s_2) \subset N(s_2)$. We easily find such an ordering by sorting the students according to their degrees, i.e., from lowest to highest degree, the students will receive labels from the smallest to the largest. Denote the resulting ordering by $\pi$. Since all neighborhoods are subsets or supersets of any other neighborhood and we sorted by degree, $\pi(s_1) \leq \pi(s_2) \Rightarrow N(s_1) \leq N(s_2)$. So we have satisfied the nested property.

 To satisfy the interval property, we order the questions according to the nesting of the neighborhoods. Recall that we have $N(\pi^{-1} (1)) \subset \cdots \subset N(\pi^{-1} (|S|))$. Now, we order the questions so that whenever $q_1 \in N(\pi^{-1}(i))$ and $q_2 \in N(\pi ^{-1}(j))$ with $i < j$, we have $q_1$ labeled smaller $q_2$ according to the ordering. We can do so by labeling the questions in $N(\pi^{-1} (1))$ the smallest numbers (the ordering within the set does not matter), then the questions in $N(\pi^{-1} (2))$ the next smallest, and so on. Call the resulting ordering $\beta$. Note that for all $s \in S$, $s=\pi^{-1}(i)$ for some $i$. So $N(s)=N(\pi^{-1}(i)) \supset N(\pi^{-1}(1))$, i.e., $s$ correctly answers the easiest question according to $\beta$. Furthermore, $N(s)$ is a block of questions that are consecutive according to the ordering $\beta$. So the interval property is also satisfied.

 To determine the run time, note that we made $O(n^2)$ comparisons of neighborhoods. Each set intersection of two neighborhoods took $O(n)$ time assuming that each neighborhood was stored as a sorted list of the questions (sorted by any fixed labeling of the questions). Ordering the students by degree took $O(n \log n)$ time and ordering the questions took $O(n)$ time. So the total run time is $O(n^2)$.
 \end{proof}

 \intervalEquivNested*
 \begin{proof}
 \label{pf:intervalEquivNested}
 First, we prove the forward direction. Assume that $G=(S \cup Q, E)$ satisfies the interval property with respect to the ordering $\beta$ on $Q$. By definition of interval property, for every $u \in S$, we have $N(u) = \{\beta^{-1}(1),\ldots, \beta^{-1}(j)\}$ for some $j \in [|Q|]$. Then for every $u_1, u_2 \in S$, we have $N(u_1) \subset N(u_2)$ or $N(u_2) \subset N(u_1)$. Let $\alpha$ be an ordering of $S$ by degree of each $u \in S$. Then the nested property holds with respect to $\alpha$.

 Second, we prove the backward direction. Assume that $G=(S \cup Q, E)$ satisfies the nested property with respect to $\alpha$ on $S$. Then $N(\alpha^{-1}(1))\subset \cdots \subset N(\alpha^{-1}(|S|))$. Using the algorithm in the proof of Proposition~\ref{prop:IMO} for IMO, we obtain an ordering $\beta$ on $Q$ so that the interval property holds with respect to $\beta$.
 \end{proof}

 \oneSide*
 \begin{proof}
 \label{pf:oneSide}
 Consider the special case that one side of the correct ordering is given to us, say the questions are given in hardest to easiest order $v_1 \geq \cdots \geq v_q$. Then we can find the minimum number of errors needed to satisfy the required conditions by correcting the edges incident to each student $u$ individually.

 We know by the interval property that every student $u$ must correctly answer either a set of consecutive questions starting from $v_1$ or no questions at all. For each $u \in S$,and for each $v_j$, simply compute the number of edge edits required so that the neighborhood of $u$ becomes $\{v_1,\ldots, v_j\}$. Select the question $v_u$ that minimizes the cost of enforcing $\{v_1, \ldots, v_j\}$ to be the neighborhood of $u$. Once the edges have been corrected, order the students by the containment relation of their neighborhoods.

 The algorithm correctly calculates the minimum edge edits since the interval property was satisfied at the minimum cost possible per student. The algorithm found the neighborhood of each student by trying at most $|Q| < n$ difficulty thresholds $v_j$, and the cost of calculation for each threshold takes $O(1)$, by using the value calculated from the previous thresholds tried. Summing over the $|S|< n$ students gives a total running time no more than $O(n^2)$.
 \end{proof}

 \subsection{$k$-near Problems}

 \subsubsection{Constrained $k$-near Addition}
 \label{sect:cka}

 \paragraph
 \noindent We use the same framework as Constrained $k$-near Editing to solve the Constrained $k$-near Addition.
 We change the definitions of the subproblem, the relevant sets, and the costs appropriately to adapt to the Addition problem.

 \constrainedAddition*
 \begin{proof}
 \label{pf:constrainedAddition}
 First, redefine $\bm{C(i,u_i,U_i,v_{j_i})}$ to be the smallest cost of adding edges incident to the weakest $i$ positions so that $u_i$ is in position $i$, $U_i$ is the set of students in the weakest $i-1$ positions, and $v_{j_i}$ is the hardest question correctly answered by the $i$ weakest students.

 The sets $\bm{P_{i,u_i}}$ and $\bm{F_{i,u_i}}$ will stay the same as before. We redefine $\bm{S_i}:=\Big\{(u_i,U_i,v_{j_i}) : u_i \in \big[ \max \{1,i-k\}, \min \{|S|,i+k\} \big], U_i \setminus \big[1,\max\{1,i-k-1\}\big] \in F_{i,u_i}, v_{j_i} \in Q \cup \{0\}, v_{j_i} \geq \max{N_G(\{u_i\}\cup U_i)} \Big\}$. Requiring that $v_{j_i}$ is at least as hard as $N_G(\{u_i\}\cup U_i)$ ensures that the final solution will satisfy the interval property with respect to the given question order. It was not needed in the Editing problem because wherever $v_{j_i}$ landed, the edges that reach questions harder than $v_{j_i}$ were deleted. The definition of $\bm{R_{i-1,u_i,U_i,v_{j_i}}}$ will stay the same as before, but using the new definition of $S_{i-1}$ from this section. Finally, the cost $\bm{c_{u_i,v_{j_i}}}$ will become the smallest number of edge additions incident to $u_i$ so that the neighborhood of $u_i$ becomes $\{1,\ldots,v_{j_i}\}$, i.e. $c_{u_i,v_{j_i}} := |\{1,\ldots,v_{j_i}\} \setminus N_G(u_i)|$.

 The recurrence relation from Constrained $k$-near Editing still applies here. However, the base cases become $C(1,u_1,U_1,v_{j_1})=|\{1,\ldots,v_{j_1}\} \setminus N_G(u_1)|$ if $v_{j_1} > 0$, and $C(1,u_1,U_1,v_{j_1})=0$ if $v_{j_1} = 0$.

 The run time is still dominated by the dynamic program since the time to construct $S_i$ becomes only $|Q|$ times larger (to enforce the additional constraint that $v_{j_i}$ is hard enough). Hence the total time to solve this problem remains $O(n^3 k^{2k+2})$.
 \end{proof}

 \subsection{Unconstrained $k$-near}

 First, we solve the Unconstrained $k$-near Addition problem in time $O(n^3 k^{2k+2})$. Second, we show that the Unconstrained $k$-near Editing problem is NP-hard.

 Assume that the students are given in $k$-near order $1,\ldots,|S|$. The questions are allowed to be ordered arbitrarily in the final solution.

 \subsubsection{Unconstrained $k$-near Addition}

 \unconstrainedAddition*
 \begin{proof}
 We introduce subproblems of the form $(i,u_i,U_i)$. Define $\bm{C(i,u_i,U_i)}$ to be the smallest number of edges incident to the weakest $i$ positions that must be added so that $u_i$ is in position $i$ and $U_i$ is the set of the $i-1$ weakest students.

 We use the same $\bm{P_{i,u_i}}$ and $\bm{F_{i,u_i}}$ as defined for Constrained $k$-near Editing to bound the search space for $U_i$ given that $u_i$ is in position $i$. Define $\bm{S_i}:=\Big\{ (u_i,U_i): u_i \in \big[\max\{1,i-k\}, \min\{|S|,i+k\} \big], U_i \setminus [1, \max\{1,i-k-1\} \in F_{i,u_i} \Big\}$.

 Next, define $\bm{R_{i-1,u_i,U_i}}:=\Big\{ (u_{i-1},U_{i-1}) \in S_{i-1}: U_i = \{u_{i-1}\} \cup U_{i-1} \Big\}$. The set $R_{i-1,u_i,U_i}$ ensures that the smaller subproblems have prefixes that are compatible with those assigned in the bigger subproblems they came from. Compatibility is illustrated in Figure~\ref{fig:unconstrainedDP}.

 \begin{figure}
 \centering
   \includegraphics[scale=.35]{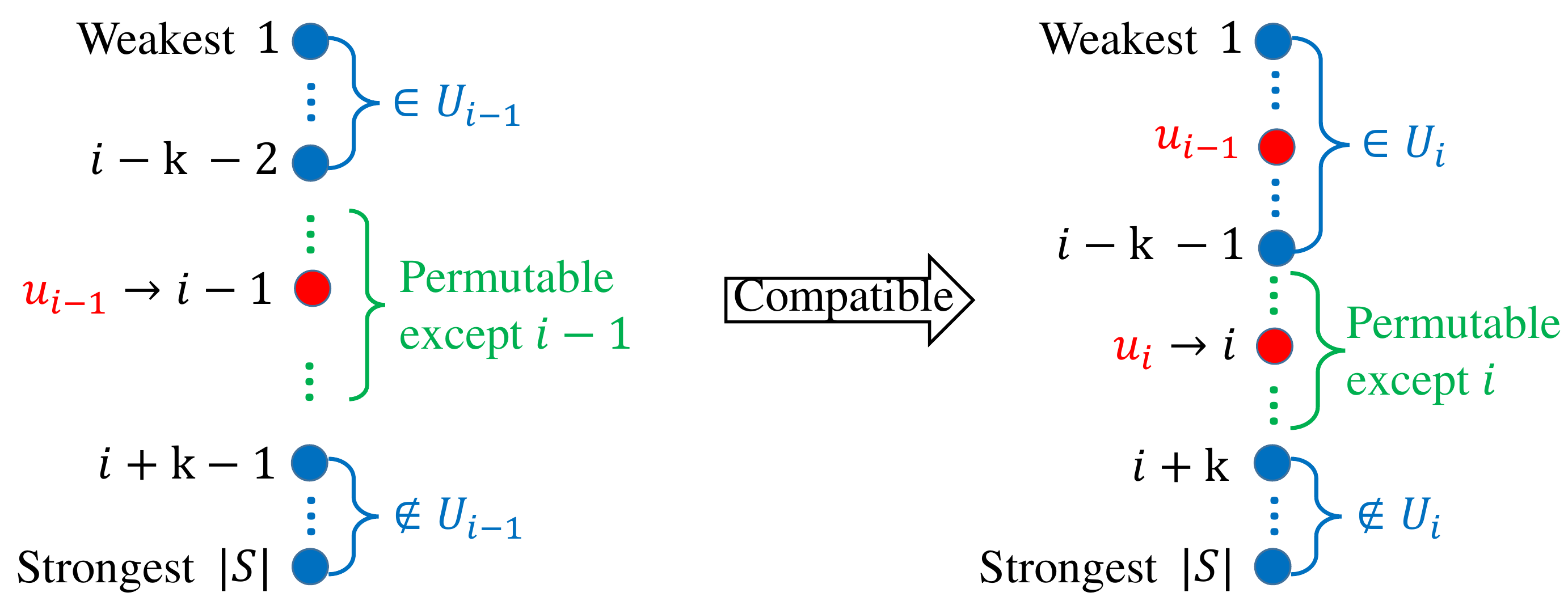}
   \caption{Subproblem $(i-1,u_{i-1},U_{i-1})$ is compatible with subproblem $(i,u_i,U_i)$ if and only if $U_i = \{u_{i-1}\} \cup U_{i-1}$. The cost of $(i,u_i,U_i)$ is sum of the minimum cost among feasible compatible subproblems of the form $(i-1,u_{i-1},U_{i-1})$ and the minimum number of additions incident to $u_i$ to make its neighborhood contain the existing neighbors of $U_i$.}
   \label{fig:unconstrainedDP}
 \end{figure}

 Lastly, define $\bm{c_{u_i,U_i}}$ to be the smallest number of edge additions incident to $u_i$ so that the neighborhood of $u_i$ will contain $N_G(U_i \cup \{u_i\})$, i.e. $c_{u_i,U_i} := |N_G(U_i \cup \{u_i\}) \setminus N_G(u_i)|$.

 Using the above definitions, we have the following recurrence.

 $$C(i,u_i,U_i)= \min_{(u_{i-1},U_{i-1}) \in R_{i-1,u_i,U_i}} \{C(i-1,u_{i-1},U_{i-1})\}+c_{u_i,U_i} $$

 The base cases are $C(1,u_1,U_1)=|N_G(U_1) \setminus N_G(u_1)|$ for all $(u_1,U_1) \in S_1$, since $u_1$ must add edges to the questions that the weaker students correctly answered.

 \sloppy The final solution to Unconstrained $k$-near Addition is $\min_{(u_{|S|},U_{|S|}) \in S_{|S|}} {C(|S|,u_{|S|},U_{|S|})}$.

 To bound the run time, note that generating $S_i$ takes $O(n \cdot k^k k^2)$ time. The dynamic program will dominate the run time again. In the dynamic program, each subproblem $(i,u_i,U_i)$ takes $O(|R_{i-1,u_i,U_i}|)$ time. So the total time is $O(\sum_{i \in S, (u_i,U_i) \in S_i} |R_{i-1,u_i,U_i}|) = O(|S| |S_i| |S_{i-1}|) = O(n(n k^k)^2)=O(n^3 k^{2k+2})$.
 \end{proof}

 \subsubsection{Unconstrained $k$-near Editing}

 \paragraph
 \noindent The Unconstrained $k$-near Editing problem is NP-hard even for $k=1$. We closely follow the proof of Drange et. al.~\cite{DDLS15} for the NP-hardness of Threshold Editing to show that Unconstrained $k$-near Editing is NP-hard. In Drange et. al.'s construction, they specified a partial order for which the cost of Threshold Editing can only worsen if the output ordering deviates from it. We crucially use this property to prove NP-hardness for Unconstrained $1$-near Editing.

 \NPhard*
 \begin{proof}
 \label{pf:NPhard}
 Let $G=(S,Q,E)$ be a bipartite graph with initial student ordering $\pi$. Consider the decision problem $\Pi$ of determining whether there is a 1-near unconstrained editing of at most $t$ edges for the instance $(G,\pi)$. We reduce from 3-SAT to $\Pi$. Let $\Phi$ be an instance for 3-SAT with clauses $C=\{c_1,\ldots,c_m\}$ and variables $V=\{v_1,\ldots,v_n\}$. We construct the corresponding instance $\Pi=(G_{\Phi},\pi_{\Phi}, t_{\Phi})$ for 1-near unconstrained editing as follows.
 First we order the variables in  an arbitrary order and use this order to define $\pi$.
 For each variable $v_i$, create six students $s^i_a, s^i_b, s^i_f, s^i_t, s^i_c, s^i_d$. Next, we define a partial ordering $P$ that the initial order $\pi_{\Phi}$ shall obey. Define $P$ to be the partial order satisfying $s^i_a > s^i_b > s^i_f, s^i_t > s^i_c > s^i_d$ for all $i \in [n]$ and $s^i_{\alpha} > s^j_{\beta}$ for all $i < j$, $\alpha, \beta \in \{a,b,c,d,f,t\}$. Define $\pi_{\Phi}$ to be the linear ordering satisfying all relations of $P$ for the variables in the initial arbitrary order, and additionally $s^i_f > s^i_t$. We remark that the proof works regardless of whether we set $s^i_f > s^i_t$ or $s^i_f < s^i_t$ in $\pi_{\Phi}$. We shall impose that optimal solutions satisfy all of the relations of $P$. To do so, for every $s > s'$, we add $t_{\Phi}+1$ new questions each with edges to $s$ and no edges to $s'$, and with edges to all $r>s$ in $\pi_{\Phi}$. Then whenever an editing solution switches the order of $s$ and $s'$, it must edit at least $t_{\Phi}+1$ edges. After adding the necessary questions to ensure feasible solutions must preserve the partial order $P$, we create a question $q_{c_l}$ for each clause $c_l$. If a variable $v_i$ appears positively in $c_l$, then add the edge $q_{c_l} s^i_t$. If $v_i$ appears negatively in $c_l$, then add the edge $q_{c_l} s^i_f$. If $v_i$ does not occur in $c_l$, then add the edge $q_{c_l} s^i_c$. For all variables $v_i$ and clauses $c_l$, add the edges $q_{c_l} s^i_b$ and $q_{c_l} s^i_d$. Finally, define $t_{\Phi} = |C| (3 |V| - 1)$. Refer to Figure~\ref{fig:NPhard} for an illustration of the construction.

 \begin{figure}
 \centering
   \includegraphics[scale=.35]{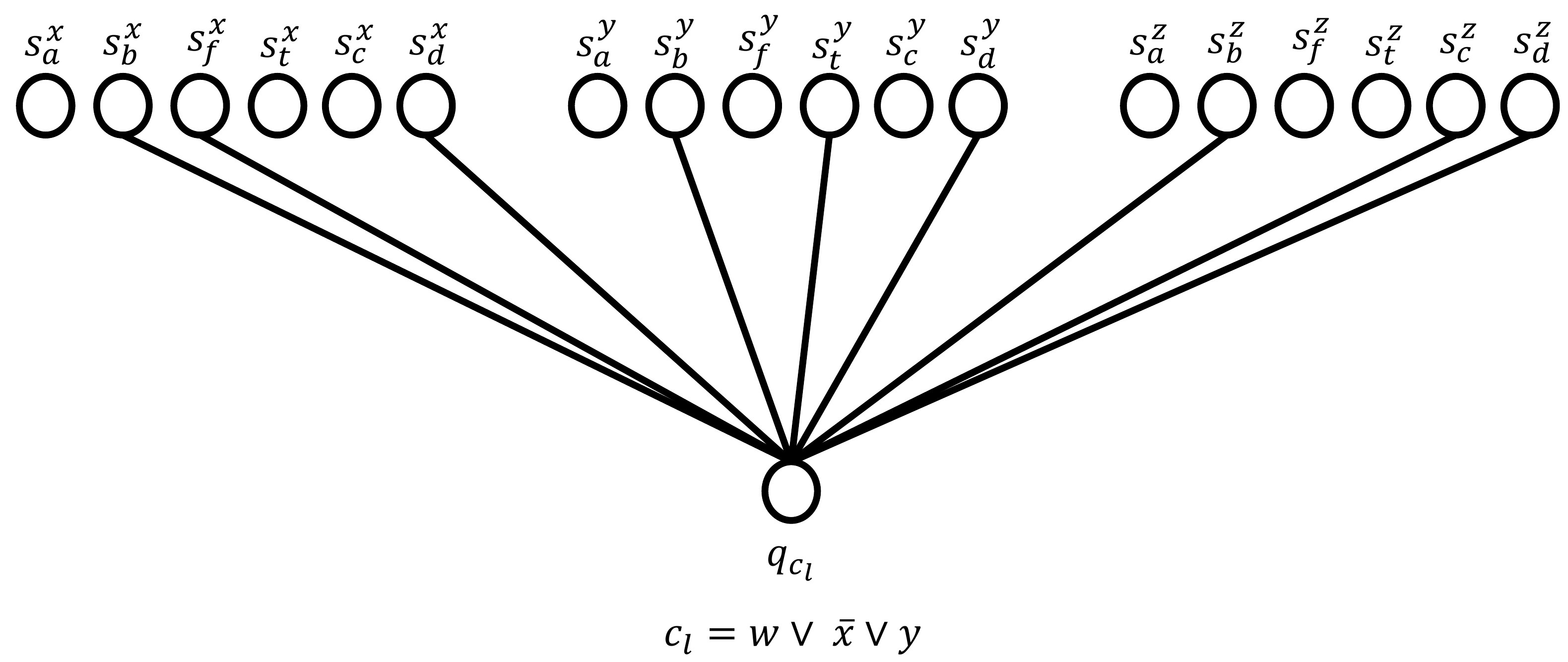}
   \caption{Each set of six vertices represent the students corresponding to a variable $x$, $y$, or $z$. The bottom vertex represents a question corresponding to the clause $c_l = w \vee \bar{x} \vee y$.}
   \label{fig:NPhard}
 \end{figure}

 Now, we show that there is a satisfying assignment if and only if there is a 1-near editing of at most $t_{\Phi}$ edges. First, we prove the forward direction. Assume there is a satisfying assignment $f:V \rightarrow \{T,F\}$. Let $c_l$ be a clause. One of the literals $v_i$ in $c_l$ is set to $T$ under the assignment $f$. If $v_i$ occurs positively, then edit the neighborhood of $q_{c_l}$ to be all students $s$ such that $s \geq s^i_t$ according to $P$ and impose $s^i_t > s^i_f$ in the solution. If $v_i$ occurs negatively in $q_{c_l}$, then edit the neighborhood of $q_{c_l}$ to be all students $s$ such that $s \geq s^i_f$ and keep the initial order that $s^i_f > s^i_t$. In both cases, the neighborhood of $q_{c_l}$ changed by 2 among the six students corresponding the variable $v_i$ and changed by 3 for the remaining groups of six students. So the number of edge edits incident to each (clause) question is $3|V|-1$. Note that the neighborhoods of the extra questions we added to impose $P$ are already nested because each time a new question was added, it received edges to all students who are stronger than a particular student according to $P$. So only the questions that came from clauses potentially need to edit their neighborhoods to achieve nesting. Hence, the total number of edge edits is $|C| (3|V|-1) = t_{\Phi}$.

 Second, we prove the backward direction. Assume there is an unconstrained 1-near editing of $|C|(3|V|-1)$ edges to obtain a chain graph. Let $c_l$ be a clause. For any variable $v_j$ not occurring in $c_l$, the original edges that $q_{c_l}$ has to the six students corresponding to $v_j$ are to $s^j_b, s^j_c, s^j_d$. If the cut-off point of the edited neighborhood of $q_{c_l}$ is among $s^j_a, s^j_b, s^j_f, s^j_t, s^j_c, s^j_d$, then the edges incident to $q_{c_l}$ must change by at least three among those six, which means that $q_{c_l}$ would have at least $3|V|$ edges incident to it. If the cut-off point of the edited neighborhood of $q_{c_l}$ is among the six students corresponding to a variable $v_i$ that occurs in $c_l$, then the edges incident to $q_{c_l}$ must change by at least two (by switching the order of $s^i_f$ and $s^i_t$ when needed) among those six students and at least three for the students corresponding to the remaining variables. Thus $q_{c_l}$ has at least $3|V|-1$ edges edits incident to it for every $c_l$. So the smallest number of edge edits possible is at least $|C|(3|V|-1)$. By the assumption, $G_{\Phi}$ has a feasible editing of at most $|C|(3|V|-1)$ edges. Then each $q_{c_l}$ must have exactly $3|V|-1$ edits incident to it. So the cut-off point for the edited neighborhood of each $q_{c_l}$ must occur among the six students corresponding to a variable $v_i$ occurring inside $c_l$. If the occurring variable $v_i$ is positive, then the cut-off point must have been at $s^i_t$ and required $s^i_t > s^i_f$ since all other cut-offs incur at least three edits. Similarly, if $v_i$ is negative, then the cut-off point must have been at $s^i_f$ and required $s^i_f > s^i_t$. All clauses must be consistent in their choice of the ordering between $s^i_f$ and $s^i_t$ for all $i \in [n]$ since the editing solution was feasible. Hence, we obtain a satisfying assignment by setting each variable $v_i$ true if and only if $s^i_t > s^i_f$.
 \end{proof}

 \subsection{Both $k$-near}
 We will solve the Both $k$-near Editing and Addition problems in time $O(n^3 k^{4k+4})$. We first show our solution for the Editing problem and then adapt it to the Addition problem.

 Assume that the students and questions are both given in $k$-near order with student labels $1,\ldots, |S|$, and question labels $1,\ldots, |Q|$.

 \subsubsection{Both $k$-near Editing}
 \label{sect:bke}

 \bothEditing*
 \begin{proof}
 \sloppy We consider subproblems of the form $(i,u_i,U_i,j_i,v_{j_i},V_{j_i})$. Define $\bm{C(i,u_i,U_i,j_i,v_{j_i},V_{j_i})}$ to be the smallest number of edges incident to the weakest $i$ students that must be edited so that student $u_i$ is in position $i$, $U_i$ is the set of the $i-1$ weakest students, $j_i$ is the position of the hardest question correctly answered by $U_i \cup \{u_i\}$, $v_{j_i}$ is the question in position $j_i$, and $V_{j_i}$ is the set of the $j_i-1$ easiest questions.

 \noindent \textbf{Feasible and Compatible Subproblems.}
 \sloppy Next, we define the search space for $(u_i,U_i,j_i,v_{j_i}, V_{j_i})$ given that $u_i$ is in position $i$. We use the same $\bm{P_{i,u_i}}$ and $\bm{F_{i,u_i}}$ defined in the proof for Constrained $k$-near Editing. Define $\bm{S_i}:=\Big\{ (u_i,U_i,j_i,v_{j_i}, V_{j_i}): u_i \in \big[ \max \{1,i-k\}, \min \{|S|,i+k\} \big], U_i \setminus \big[1,\max\{1,i-k-1\}\big] \in F_{i,u_i}, v_{j_i} \in \big[ \max\{1,j_i-k\}, \min\{|Q|,j_i+k\} \big], V_{j_i} \setminus \big[ 1, \max\{1,j_i-k-1\}\big] \in F_{j_i,v_{j_i}} \Big\}$. Here, we need to constrain both the student side and the question side to make sure that all elements are $k$-near as opposed to only enforcing the $k$-nearness on the students in Constrained $k$-near Editing.

 \sloppy To bound the search space for subproblems to be compatible with the bigger subproblems they came from, we define $\bm{R_{i-1,u_i,U_i,j_i,v_{j_i},V_{j_i}}} := \Big\{ (u_{i-1},U_{i-1},j_{i-1},v_{j_{i-1}}, V_{j_{i-1}}) \in S_{i-1}: U_i = U_{i-1} \cup \{u_{i-1}\}, j_i \geq j_{i-1}, V_{j_i} \cup \{v_{j_i}\} \supset V_{j_{i-1}} \cup \{v_{j_{i-1}}\}, j_i > j_{i-1} \Rightarrow V_{j_i} \supset V_{j_{i-1}} \cup \{v_{j_{i-1}}\} \Big\}$. The constraints in the set $R_{i-1,u_i,U_i,j_i,v_{j_i},V_{j_i}}$ ensure that the prefixes of position $i$ and position $j_i$ in the smaller subproblem will be compatible with the bigger subproblem that it came from. Furthermore, $j_i \geq j_{i-1}$ ensures that stronger students correctly answer all questions that weaker students correctly answered. We demonstrate compatibility in Figure~\ref{fig:bothNearDP}.

 \begin{figure}
   \includegraphics[scale=.45]{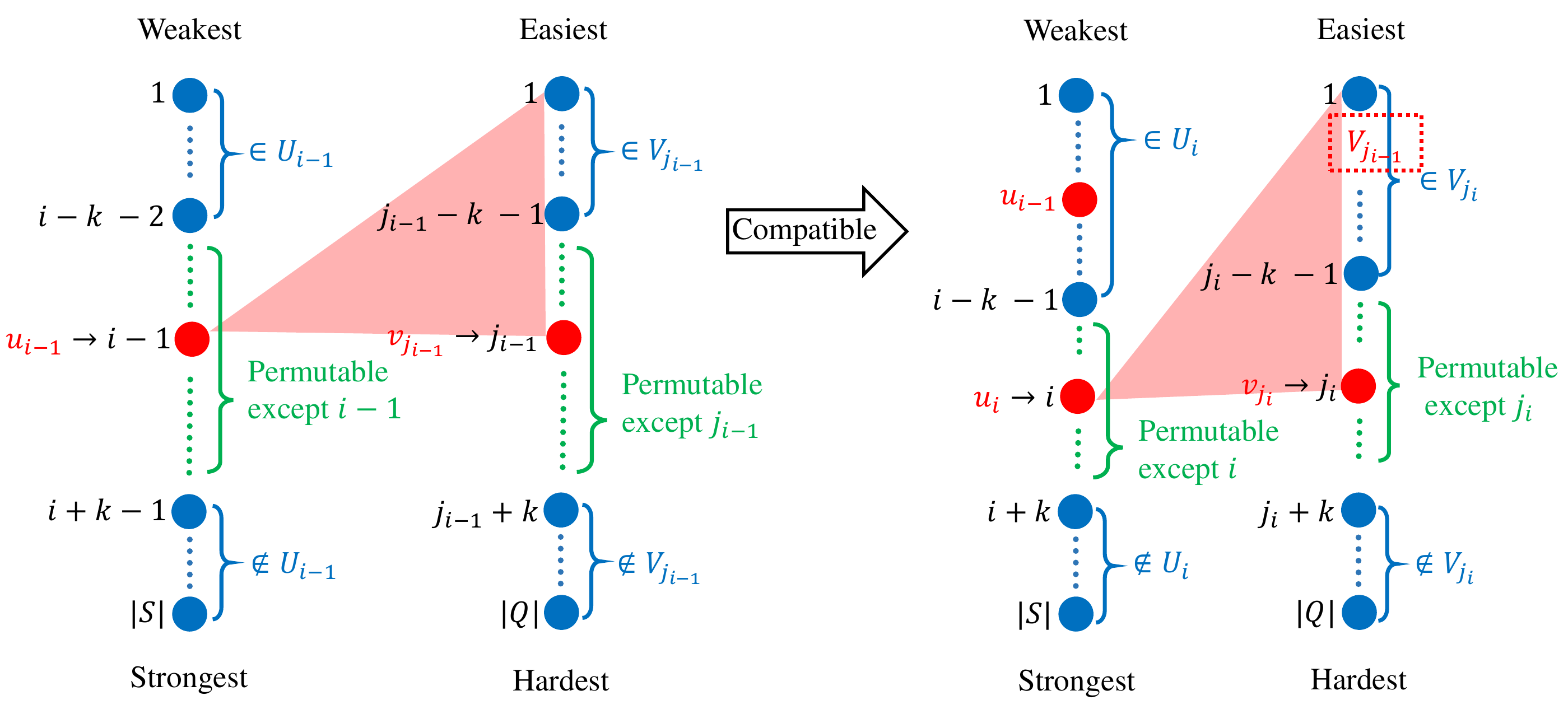}
   \caption{Subproblem $(i-1, u_{i-1}, U_{i-1}, j_{i-1}, v_{j_{i-1}}, V_{j_{i-1}})$ is compatible with subproblem $(i, u_i, U_i, j_i, v_{j_i}, V_{j_i})$ if and only if $U_i = \{u_i\} \cup U_{i-1}$, $j_{i-1}$ represents a position no harder than $j_i$, $V_{j_i} \cup \{v_{j_i}\}$ contains $V_{j_{i-1}} \cup \{v_{j_{i-1}}\}$, and $j_{i-1}$ strictly easier than $j_i$ implies that $V_{j_i}$ contains $V_{j_{i-1}} \cup \{v_{j_{i-1}}\}$. The cost of $(i, u_i, U_i, j_i, v_{j_i}, V_{j_i})$ is the sum of the minimum cost among feasible compatible states of the form $(i-1, u_{i-1}, U_{i-1}, j_{i-1}, v_{j_{i-1}}, V_{j_{i-1}})$ and the minimum number of edits incident to $u_i$ that makes its neighborhood $V_{j_i} \cup \{v_{j_i}\}$.}
   \label{fig:bothNearDP}
 \end{figure}

 \noindent \textbf{The Dynamic Program.}
 Finally, define $\bm{c_{u_i,v_{j_i},V_{j_i}}}$ to be the smallest number of edge edits incident to $u_i$ so that the neighborhood of $u_i$ becomes exactly $V_{j_i} \cup \{v_{j_i}\}$, i.e. $c_{u_i,v_{j_i},V_{j_i}} := |N_G(u_i) \triangle V_{j_i} \cup \{v_{j_i}\}|$.

 Using the above definitions, we obtain the following recurrence.

 {\setlength{\mathindent}{0cm}
 \begin{align*}
 & C(i,u_i,U_i,j_i,v_{j_i},V_{j_i})=\\
 & \min_{(u_{i-1},U_{i-1},j_{i-1},v_{j_{i-1}}, V_{j_{i-1}}) \in R_{i-1,u_i,U_i,j_i,v_{j_i},V_{j_i}}}  \{C(i-1,u_{i-1},U_{i-1},j_{i-1},v_{j_{i-1}}, V_{j_{i-1}}) \}\\
 & \quad + c_{u_i,v_{j_i},V_{j_i}}
 \end{align*}
 }

 The base cases are $C(1,u_1,U_1,j_1,v_{j_1},V_{j_1})=|N_G(u_1) \triangle \{v_{j_1}\} \cup V_{j_1}|$ for all $(u_1,U_1,j_1,v_{j_1},V_{j_1}) \in S_1$.

 \sloppy The final solution is $\min_{(u_{|S|},U_{|S|},j_{|S|},v_{j_{|S|}},V_{j_{|S|}}) \in S_{|S|}} {C(|S|,u_{|S|},U_{|S|},j_{|S|},v_{j_{|S|}},V_{j_{|S|}}})$.

 \noindent \textbf{Running Time.}
 \sloppy First, observe that $|S_i| = O(k^2 k^{2k} |Q|)$, since there are $O(k)$ choices for $u_i$ and $v_i$, $O(k^k)$ choices for $U_i$ and $V_{j_i}$, and $|Q|$ choices for $j_i$. To build $S_i$, we need to build $F_{i,u_i}$ and $F_{j_i, v_{j_i}}$. In Section~\ref{sect:exacts}, we saw that each of the $F_{i,u_i}$ takes $O(k^2 k^k)$ time to build. Then building the set $S_i$ is upper bounded by $O(k \cdot k^k k^2 \cdot |Q| \cdot k \cdot k^k k^2)$ per $i$, where we are over-counting the time to generate all possible $U_i$ and $V_{j_i}$ by the time it takes to build $F_{i,u_i}$ and $F_{j_i,v_{j_i}}$. Building the set $R_{i-1,u_i,U_i,j_i,v_{j_i},V_{j_i}}$ while building $S_i$ will take $O(|S|+|Q|)$ to check the conditions that restrict $S_{i-1}$ to $R_{i-1,u_i,U_i,j_i,v_{j_i},V_{j_i}}$. Due to the size of $S_i$, the construction of sets will still be dominated by the time to solve the dynamic program. Specifically, each subproblem $(i,u_i,U_i,j_i,v_{j_i},V_{j_i})$ takes $O(|R_{i-1,u_i,U_i,j_i,v_{j_i},V_{j_i}}|)$ time. So the total time is $O(\sum_{i \in S, (u_i,U_i,j_i,v_{j_i},V_{j_i}) \in S_i} |R_{i-1,u_i,U_i,j_i,v_{j_i},V_{j_i}}|)=O(|S| |S_i| |S_{i-1}|) = O(n(k^2 \cdot k^{2k} n)^2)= O(n^3 k^{4k+4})$.
 \end{proof}

 \subsubsection{Both $k$-near Addition}

 \paragraph
 \noindent To solve the Addition version, we apply the method from the solution for Both $k$-near Editing.

 \bothAddition*
 \begin{proof}
 We redefine $\bm{C(i,u_i,U_i,j_i,v_{j_i},V_{j_i})}$ to be the smallest number of edges incident to the weakest $i$ students that must be added so that student $u_i$ is in position $i$, $U_i$ is the set of the $i-1$ weakest students, $j_i$ is the position of the hardest question correctly answered by $U_i \cup \{u_i\}$, $v_{j_i}$ is the question in position $j_i$, and $V_{j_i}$ is the set of the $j_i-1$ easiest questions.

 We keep $\bm{P_{i,u_i}}$ and $\bm{F_{i,u_i}}$ the same as in the proof for Constrained $k$-near Editing. Redefine $\bm{S_i}:=\Big\{ (u_i,U_i,j_i,v_{j_i}, V_{j_i}): u_i \in \big[ \max \{1,i-k\}, \min \{|S|,i+k\} \big], U_i \setminus \big[1,\max\{1,i-k-1\}\big] \in F_{i,u_i}, v_{j_i} \in \big[ \max\{1,j_i-k\}, \min\{|Q|,j_i+k\} \big], V_{j_i} \setminus \big[ 1, \max\{1,j_i-k-1\}\big] \in F_{j_i,v_{j_i}}, V_{j_i}\cup \{v_{j_i}\} \supset N_G(\{u_i\} \cup U_i) \Big\}$. The addition constraint $V_{j_i}\cup \{v_{j_i}\} \supset N_G(\{u_i\} \cup U_i)$ is added here to ensure that the interval property induced by the current student ordering is satisfied every step. It was not needed in section~\ref{sect:bke} because existing edges to questions outside $V_{j_i} \cup \{v_{j_i}\}$ could be deleted. The definition of $\bm{R_{i-1,u_i,U_i,j_i,v_{j_i},V_{j_i}}}$ remains the same as section~\ref{sect:bke}, but using the newly defined $S_{i-1}$. Lastly, redefine $\bm{c_{u_i,v_{j_i},V_{j_i}}}$ to be the smallest number of edge additions incident to $u_i$ so that the neighborhood of $u_i$ becomes exactly $V_{j_i} \cup \{v_{j_i}\}$, i.e. $c_{u_i,v_{j_i},V_{j_i}} := |V_{j_i} \cup \{v_{j_i}\} \setminus N_G(u_i)|$.

 The general recurrence relation of Section~\ref{sect:bke} stays the same. The base cases change to $C(1,u_1,U_1,j_1,v_{j_1},V_{j_1})=|\{v_{j_1}\}\cup V_{j_1} \setminus N_G(u_1)|$, with the convention that $j_1=0$ means $V_{j_1}=\emptyset$ and $v_{j_1}$ is omitted from the count $|\{v_{j_1}\}\cup V_{j_1}|$.

 Although the time to construct $S_i$ is larger by a factor of $|Q|$, the total run time is dominated by the dynamic program, which takes $O(n^3 k^{4k+4})$.
 \end{proof}

\end{document}